\documentclass[runningheads]{llncs}
\usepackage{amssymb,amsmath}
\usepackage{graphics,graphicx,color}
\usepackage{enumerate,paralist,hyperref}
\usepackage[labelsep=period,labelfont=bf,size=small,compatibility=false]{caption}
\usepackage{subfig,xspace}
\usepackage{setspace,todonotes}
\usepackage{array}
\newcommand{\ver}{arxiv}
\newcommand{\arxapp}[2]{\ifthenelse{\equal{\ver}{conf}}{#2}{#1}}
\graphicspath{{images/}}

\let\doendproof\endproof
\renewcommand\endproof{~\hfill$\qed$\doendproof}

\newcommand{\rac}[1]{\ensuremath{#1}-bend RAC\xspace}
\newcommand{\good}{good\xspace}

\begin{document}
\title{On RAC Drawings of Graphs with~one~Bend~per~Edge}
\author{Patrizio Angelini, Michael~A.~Bekos, Henry~F\"orster, Michael~Kaufmann} 
\authorrunning{P.~Angelini et al.}

\institute{%
Wilhelm-Schickhard-Institut f\"ur Informatik, Universit\"at T\"ubingen, Germany\\
\texttt{$\{$angelini,bekos,foersth,mk$\}$@informatik.uni-tuebingen.de}
}


\maketitle

\begin{abstract}
	A \emph{$k$-bend right-angle-crossing drawing} (or \emph{\rac{k} drawing}, for short) of a graph is a polyline drawing where each edge has at most $k$ bends and the angles formed at the crossing points of the edges are~$90^\circ$. Accordingly, a graph that admits a \rac{k} drawing is referred to as \emph{$k$-bend right-angle-crossing graph} (or \emph{\rac{k}}, for short).
	
	In this paper, we continue the study of the maximum edge-density of \rac{1} graphs. We show that an $n$-vertex \rac{1} graph cannot have more than $5.5n-O(1)$ edges. We also demonstrate that there exist infinitely many $n$-vertex \rac{1} graphs with exactly $5n-O(1)$ edges. Our results improve both the previously known best upper bound of $6.5n-O(1)$ edges and the corresponding lower bound of $4.5n-O(\sqrt{n})$ edges by Arikushi et al. (Comput. Geom. 45(4), 169--177 (2012)).
\end{abstract}

\section{Introduction}
\label{sec:introduction}

A recent research direction in Graph Drawing, which is currently receiving a great deal of attention~\cite{Shonan2016,Dagstuhl2016,SoCG2017}, focuses on combinatorial and algorithmic aspects for families of graphs that can be drawn on the plane while avoiding specific kinds of edge crossings; see, e.g.,~\cite{DBLP:journals/corr/abs-1804-07257} for a survey. This direction is informally recognized under the term ``beyond planarity''. An early work on beyond planarity (and probably the one that initiated this direction in Graph Drawing) is due to Didimo, Eades, and Liotta~\cite{DBLP:journals/tcs/DidimoEL11}, who introduced and first studied the family of graphs that admit polyline drawings, with few bends per edge, in which the angles formed at the edge crossings are $90^\circ$. Their primary motivation stemmed from experiments indicating that the humans' abilities to read and understand drawings of graphs are not affected too much, when the edges cross at large angles~\cite{DBLP:conf/apvis/Huang07,DBLP:journals/vlc/HuangEH14} and the number of bends per edge is limited~\cite{DBLP:journals/iwc/Purchase00,DBLP:journals/ese/PurchaseCA02}. Their work naturally gave rise to a systematic study of several different variants of these graphs; see, e.g.,~\cite{DBLP:journals/jgaa/AngeliniCDFBKS11,DBLP:journals/jgaa/ArgyriouBS12,DBLP:journals/comgeo/ArikushiFKMT12,DBLP:journals/tcs/BekosDLMM17,DBLP:journals/ipl/DidimoEL10,DBLP:journals/dam/EadesL13,DBLP:journals/algorithmica/GiacomoDEL14,DBLP:journals/mst/GiacomoDLM11}.

Formally, a \emph{$k$-bend right-angle-crossing drawing} (or \emph{\rac{k} drawing}, for short) of a graph is a polyline drawing where each edge has at most $k$ bends and the angles formed at the crossing points of the edges are $90^\circ$. Accordingly, a graph that admits a \rac{k} drawing is referred to as \emph{$k$-bend right-angle-crossing graph} (or \emph{\rac{k}}, for short); a \rac{0} graph (drawing) is also called a \emph{straight-line RAC graph} (\emph{drawing}). 

There exist several results for straight-line RAC graphs. Didimo et al.~\cite{DBLP:journals/tcs/DidimoEL11} showed that a straight-line RAC graph with $n$ vertices has at most $4n-10$~edges, which is a tight bound, i.e., there exist infinitely many straight-line RAC graphs with $n$ vertices and exactly $4n-10$ edges. These graphs are actually referred to as \emph{optimal} or \emph{maximally-dense} straight-line RAC and are in fact $1$-planar~\cite{DBLP:journals/dam/EadesL13}, i.e., they admit drawings in which each edge is crossed at most once. In general, however, deciding whether a graph is straight-line RAC is NP-hard~\cite{DBLP:journals/jgaa/ArgyriouBS12}, and remains NP-hard even if the drawing must be upward~\cite{DBLP:journals/jgaa/AngeliniCDFBKS11} or $1$-planar~\cite{DBLP:journals/tcs/BekosDLMM17}. Bachmaier et al.~\cite{DBLP:journals/dam/BachmaierBHNR17} and Brandenburg et al.~\cite{DBLP:journals/tcs/BrandenburgDEKL16} presented interesting relationships between the class of straight-line RAC graphs and subclasses of $1$-planar graphs. Variants, in which the vertices are restricted on two parallel lines or on a circle, have been studied by Di~Giacomo et al.~\cite{DBLP:journals/algorithmica/GiacomoDEL14}, and by Hong and Nagamochi~\cite{DBLP:conf/wg/HongN15}. 

An immediate observation emerging from this short literature overview is that the focus has been primarily on the straight-line case; the results for RAC drawings with bends are significantly fewer. Didimo et al.~\cite{DBLP:journals/tcs/DidimoEL11} observed that $1$- and \rac{2} graphs have a sub-quadratic number of edges, while any graph with $n$ vertices admits a \rac{3} drawing in $O(n^4)$ area; the required area was improved to $O(n^3)$ by Di~Giacomo et al.~\cite{DBLP:journals/mst/GiacomoDLM11}. 
Quadratic area for \rac{1} drawings can be achieved for subclasses of $1$-plane graphs~\cite{clwz-brdngqa-eurocg18}; for general $1$-plane graphs the known algorithm may yield \rac{1} drawings with super-polynomial area~\cite{DBLP:journals/tcs/BekosDLMM17}. The best-known upper bounds on the number of edges of $1$- and \rac{2} graphs are due to Arikushi et al.~\cite{DBLP:journals/comgeo/ArikushiFKMT12}, who showed that these graphs can have at most $6.5n-13$ and $74.2n$ edges, respectively. Arikushi et al.~\cite{DBLP:journals/comgeo/ArikushiFKMT12} also presented $1$- and \rac{2} graphs with $n$ vertices, and $4.5n-O(\sqrt{n})$ and $7.83n-O(\sqrt{n})$ edges, respectively. Angelini et al.~\cite{DBLP:journals/jgaa/AngeliniCDFBKS11} have shown that all graphs with maximum vertex degree $3$ are \rac{1}, while those with maximum vertex degree~$6$ are \rac{2}. It is worth noting that the complexity of deciding whether a graph is $1$- or \rac{2} is still open. 

\medskip\noindent\textit{Our contribution:} In this work, we present improved lower and upper bounds on the maximum edge-density of \rac{1} graphs. Note that this type of problems is commonly referred to as Tur\'an type, and has been widely studied also in the framework of beyond planarity; see, e.g.,~\cite{Ackerman09,DBLP:journals/corr/Ackerman15,DBLP:journals/jgaa/AckermanKV18,AckermanT07,AgarwalAPPS97,DBLP:conf/gd/Bekos0R16,CheongHKK15,FoxPS13,DBLP:journals/corr/KaufmannU14,PachRTT06,PachT97,MR0187232}. More precisely, in Section~\ref{sec:upper-bound}, we show that an $n$-vertex \rac{1} graph cannot have more than $5.5n-O(1)$ edges, while in Section~\ref{sec:lower-bound} we demonstrate that there exist infinitely many \rac{1} graphs with $n$ vertices and exactly $5n-O(1)$ edges. These two results together further narrow the gap between the best-known lower and upper bounds on the maximum edge-density of \rac{1} graphs (from $2n$ to $n/2$). Our approach for proving the upper bound in Section~\ref{sec:upper-bound} builds upon the charging technique by Arikushi et al.~\cite{DBLP:journals/comgeo/ArikushiFKMT12}, which we overview in Section~\ref{sec:preliminaries}. We discuss open problems in Section~\ref{sec:conclusions}.

\section{Overview of the Charging Technique}
\label{sec:preliminaries}

In this section, we introduce the necessary notation and we describe the most important aspects of the charging technique by Arikushi et al.~\cite{DBLP:journals/comgeo/ArikushiFKMT12} for bounding the maximum number of edges of a \rac{1} graph. Consider an $n$-vertex \rac{1} graph $G=(V,E)$, together with a corresponding \rac{1} drawing $\Gamma$ with the minimum number of crossings. The edges of $G$ are partitioned into two sets $E_0$ and $E_1$, based on whether they are crossing-free in $\Gamma$ (set $E_0$) or they have at least a crossing (set $E_1$). Let $G_0$ and $G_1$ be the subgraphs of $G$ induced by $E_0$ and $E_1$, respectively. 

Since $G_0$ is plane, $|E_0| \leq 3n-6$ holds. To estimate $|E_1|$, Arikushi et al.\ consider the graph $G_1'$ that is obtained from the drawing of $G_1$, by replacing each crossing point with a dummy vertex; we call $G_1^\prime$ the \emph{planarization} of the drawing of $G_1$. Let $V_1'$, $E_1'$, and $F_1'$ be the set of vertices, edges, and faces of $G_1'$, respectively. Let $\deg(v)$ be the degree of a vertex $v$ of $G_1'$ and $s(f)$ be the size of a face $f$ of $G_1'$, that is, the number of edges incident to $f$. In the charging scheme, every vertex $v$ of $G_1'$ is initially assigned a charge $ch(v)$ equal to $\deg(v) - 4$, while every face $f$ of $G_1'$ is initially assigned a charge $ch(f)$ equal to $s(f) - 4$. By Euler's formula, the sum of charges over all vertices and faces of $G_1'$ is:
%
\[
\sum_{v \in V_1'} (\deg(v)-4) + \sum_{f \in F_1'} (s(f) - 4) = 2|E_1'| - 4 |V_1'| + 2|E_1'| - 4 |F_1'| = -8
\]

In two subsequent discharging phases, they redistribute the charges in $G_1'$ so that  
\begin{inparaenum}[(i)]
	\item the total charge remains the same, and 
	\item all faces have non-negative charges. 
\end{inparaenum}
In the first discharging phase, for every edge $e$ with one bend, half a unit of charge is passed from each of its two endvertices to the face that is incident to the convex bend of $e$. Arikushi et al. show that each face of size less than $4$ has at least one convex bend, so it receives at least one unit of charge. Hence, after this phase, the only faces that have negative charges are the so-called \emph{lenses}, which have size $2$ and only one convex bend (each lens has charge  $-1$). On the other hand, the charge of every vertex $v \in V_1'$ is at least $ch^\prime(v)=\frac{1}{2}\deg(v)-4$. 

In the second discharging phase, Arikushi et al.\ exploit the crossing minimality of $\Gamma$ to guarantee the existence of an injective mapping from the lenses to the convex bends incident to faces of $G_1'$ with size at least $4$. Since each such bend yields one additional unit of charge to its incident face, and since this face has already a non-negative charge due to its size, it is possible to move this unit from the face to the mapped lens without introducing faces with negative charge. Hence, after the second phase, the charge $ch''(f)$ of each face $f \in F_1'$ is non-negative (and at least as large as its initial charge, i.e., $ch''(f)\geq ch(f)$). Since $ch''(v) = ch(v)$, $|E_1| \leq 4n-8$ can be proved as follows:
\begin{equation}
|E_1|-4n = \sum_{v \in V_1'} \left( \frac{1}{2}\deg(v)-4 \right) \leq 
\sum_{v \in V_1'} ch'' (v) \leq 
\sum_{v \in V_1'} ch'' (v) + \sum_{f \in F_1'} ch'' (f) = -8
\label{eq:charges}
\end{equation}

So far, graph $G$ has $|E_0| + |E_1| \leq 7n-14$ edges. Arikushi et al.\ improve this bound in a conclusive analysis based on the observation that a triangular face of $G_0$ cannot contain edges of $E_1$. Hence, if $G_0$ contains exactly $3n-6$ edges, then it is a triangulation, and thus $E_1=\emptyset$.
More in general, they considered how many edges $E_1$ may contain when $G_0$ is a graph obtained from a triangulation by removing $k$ edges. Let $V_0$, $E_0$, and $F_0$ be the sets of vertices, edges, and faces of $G_0$, respectively, and let $d(f)$ be the \emph{degree} of a face $f \in F_0$, i.e., the number of its distinct~vertices. 
Then, by Eq.~\ref{eq:charges} we have:
\begin{equation}
|E_1| \leq \sum_{f \in F_0; d(f) > 3}(4d(f)-8)
\label{eq:crossed}
\end{equation}

Arikushi et al.\ proved that the right-hand side of Eq.~\ref{eq:crossed} is at most~$8k$. In fact, the removal of any crossing-free edge $e$ leads to one of the following~cases.
\begin{enumerate}[C.1]
	\item \label{C.arikushi1} if $e$ was a bridge of a face, this yields a face with the same degree, which leaves the right-hand side of Eq.~\ref{eq:crossed} unchanged;
	\item \label{C.arikushi2} if $e$ was adjacent to two triangles, this yields a new face $f$ of degree $d(f) = 4$, which can contain at most $4d(f) - 8 = 8$ edges of $E_1$, which increases the right-hand side of Eq.~\ref{eq:crossed} by $8$;
	\item \label{C.arikushi3} if~$e$ was adjacent to a triangle and to a face of degree $d(f)$ (containing at most $4d(f)-8$ edges of $E_1$), this yields a new face of degree at most $d(f)+1$, which can contain at most $4(d(f)+1)-8=4d(f)-4$ edges of $E_1$, which increases the right-hand side of Eq.~\ref{eq:crossed} by at most $4$; finally,
	\item \label{C.arikushi4} if~$e$ was adjacent to two faces $f_1$ and $f_2$ such that $d(f_1),d(f_2) > 3$ (containing at most $4(d(f_1)+d(f_2))-16$  edges of $E_1$), this yields a new face of degree at most $d(f_1)+d(f_2)-2$, which contains at most $4(d(f_1)+d(f_2)-2)-8=4(d(f_1)+d(f_2))-16$ edges of $E_1$, leaving the right-hand side of Eq.~\ref{eq:crossed}~as is.
\end{enumerate}
Hence, the removal of $k$ uncrossed edges increases the right-hand side of Eq.~\ref{eq:crossed} by at most $8k$.
With this observation, Arikushi et al.\ derived two different upper bounds on the number of edges of $G$, namely:
\begin{eqnarray}
\label{eq:bound-1} |E| \leq (3n-6-k) + 4n-8 = 7n-14-k \\
\label{eq:bound-2} |E| \leq (3n-6-k) + 8k 
\end{eqnarray}
The minimum of the two bounds is maximized when $k=n/2-1$, which yields $|E| \leq 6.5n-13$. Arikushi et al.\ noticed that the bound of $8k$ is an overestimation, and that possible refinements would lead to improvements of the overall bound.

\section{An Improved Upper Bound}
\label{sec:upper-bound}

In this section, we describe how to improve the analysis of the charging scheme described in Section~\ref{sec:preliminaries} to obtain a better upper bound. W.l.o.g., we assume that $G$ is connected and that $n \geq 5$. Let $f$ be a face of $G_0$. As in the previous section, we denote by $d(f)$ the degree of $f$, that is, the number of distinct vertices of $f$. Since $f$ is not necessarily simple or connected, the boundary of $f$ is a disjoint set of (not necessarily simple) cycles, which are called \emph{facial walks}; see Fig.~\ref{fig:G0}. We denote by $\ell(f)$ the length of face $f$, that is, the number of edges (counted with multiplicities) in all facial walks of $f$.

Since a vertex $v$ may occur more than once in a facial walk of $f$, we denote by $m_f(v)$ the number of its occurrences in $f$ minus one (that is, the number of extra occurrences beyond its first). The sum of such extra occurrences over all the vertices of face $f$ is denoted by $m(f)$, that is, $m(f) = \sum_{v \in f}m_f(v)$. Further, we denote by $b(f)$ the number of biconnected components of all facial walks of $f$. Finally, we assume that an isolated vertex of $f$ (if any) is not a biconnected component of $f$, and we denote by $i(f)$ the number of isolated vertices of $f$. It is not difficult to see that $\ell(f) = d(f)+ m(f) - i(f)$. 

\begin{figure}[t]
	\centering
	\subfloat[\label{fig:G0}{}]
	{\includegraphics[scale=0.8,page=1]{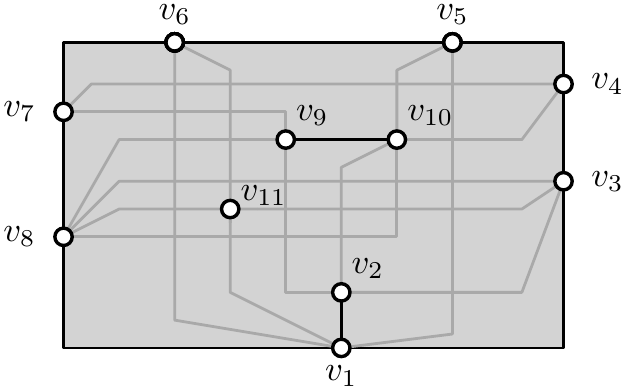}}
	\hfil 	
	\subfloat[\label{fig:G1}{}]
	{\includegraphics[scale=0.8,page=2]{faces}}
	\caption{
	(a)~Illustration of a non-simple, non-connected face $f$ of $G_0$ (colored in black). 
	The edges of $G_1$ are colored gray. 
	Face $f$ consists of two facial walks ($w_1=\langle v_1, v_2, v_1, v_3, v_4, v_5, v_6, v_7, v_8 \rangle$ and 
	$w_2=\langle v_9, v_{10}\rangle$)
	and an isolated vertex ($v_{11}$).
	Observe that $d(f)=11$ (as $f$ contains $11$ distinct vertices), 
	$\ell(f)=11$ (as the sum of the lengths of $w_1$ and $w_2$ is $11$),
	$m_f(v_1)=1$ (as $v_1$ appears twice in $w_1$), 
	$i(f)=1$ (as $v_{11}$ is an isolated vertex of $f$), and
	$b(f)=3$ (as $w_1$ consists of two biconnected components, while $w_2$ is biconnected).
	Face $f$ is \good, since each of its edges is \good. Note that removing edge $(v_4,v_7)$ would make edges $(v_5,v_6)$ and $(v_9,v_{10})$ not \good.
	(b)~The faces of $F_1^\prime(f)$ that are surrounding the three biconnected components of $f$ 
	are tiled in gray.
	}
	\label{fig:face-properties}
\end{figure}

Let $G^\prime$ be the planarization of the drawing $\Gamma$ of $G$. As opposed to $G_0$, whose faces are not necessarily connected, the faces of $G^\prime$ are in fact connected, since~$G$ is connected. Let $f$ be a face of $G_0$ and let $e$ be any edge incident to $f$. We say  that edge $e$ is \emph{\good} for $f$ if and only if there is no other edge $e^\prime$ incident to $f$ such that $e$ and $e^\prime$ are both incident to a face $g$ of $G^\prime$ that lies inside $f$. Accordingly, face $f$ is called \emph{\good} if and only if either all its edges are \good for $f$ or if $f$ is a triangle; see Fig.~\ref{fig:G0}. Note that, if each face of $G_0$ is good, then every face of the planarization $G^\prime$ is either a triangle of crossing-free edges or contains at most one crossing-free edge, and vice versa. In the next two lemmas, we assume that the faces of $G_0$ are \good; we show later how to guarantee this property. For this, we may need to introduce parallel edges (but no self-loops) in $G_0$, which however are non-homotopic (each region they define contains at least a vertex). Further, we may need to introduce planar edges with more than one bend; this does not affect the discharging scheme of Arikushi et al. which only considers $G_1$. 

\begin{lemma}\label{lem:crossed-in-f}Let $\Gamma$ be a drawing of $G$ such that all faces of $G_0$ are \good. Then, each face $f$ of $G_0$ contains at most $2d(f) -2 m(f) + 2i(f) +4b(f)- 8$ edges of~$G_1$.
\end{lemma}

\begin{proof} 
Consider the subgraph $G(f)$ of $G$ which is induced by the interior of $f$ and let $\Gamma(f)$ be the drawing of $G(f)$ derived from $\Gamma$.
We denote by $G_1(f)=(V_1(f),E_1(f))$ the subgraph of $G(f)$ induced by the set of crossing edges in $\Gamma(f)$, and by $G_1'(f)$ the planarization of $G_1(f)$.

 Let $B(f)$ be the set of biconnected components of $f$ and $F_1^\prime(f)$ the set of faces of the drawing of $G_1'(f)$ that is derived from $\Gamma(f)$. Since every edge of $f$ is \good, every biconnected component $c \in B(f)$ with length $\ell(c)$ will be surrounded by a face $f^\prime_c \in F_1^\prime(f)$ in $G^\prime_1$ that is of length $\ell(f^\prime_c) \geq 2\ell(c)$; see Fig.~\ref{fig:G1}. Hence, before the discharging phases in the charging scheme of Arikushi et al. (applied on $G_1'(f)$), the charge of face $f^\prime_c$ is at least $2\ell(c)-4$. Since after the second discharging phase, the charge of each face is at least as much as its initial charge, it follows that the charge of face $f^\prime_c$ is still at least $2\ell(c)-4$ even after the discharging phases. Since isolated vertices of $f$ are not surrounded by a face of $F_1^\prime(f)$, summing up the charges of all biconnected components of $f$, we get that 
\[
\sum \limits_{c \in B(f)} ch''(f^\prime_c) \geq  \sum \limits_{c \in B(f)} (2\ell(c)-4) = 2\ell(f) -4b(f) = 2(d(f)+m(f)-i(f)) - 4b(f)
\]
Since, after the second discharging phase, each face has a non-negative charge and the sum of the charges of faces surrounding biconnected components of $f$ is a lower bound for the sum of the charges of all faces in $F^\prime_1(f)$, we get that
\[
\sum \limits_{f^\prime \in F^\prime_1(f)} ch''(f^\prime) - \sum \limits_{c \in B(f)} ch''(f^\prime_c) \geq 0
\]
Hence, by refining Eq.~\ref{eq:charges} we obtain that the number of crossing edges in $G(f)$ can be upper-bounded as follows 
\begin{align*}
|E_1(f)| -4 d(f) & = \sum_{v \in f} \left( \frac{1}{2}\deg(v)-4 \right) \leq \sum \limits_{v \in f} ch''(v) \\
  & \leq \sum \limits_{v \in f} ch'' (v) + \sum \limits_{f^\prime \in F^\prime_1(f)} ch''(f^\prime) - 2(d(f)-m(f)+i(f)) + 4b(f) \\ 
 &  =  -8 - 2(d(f) + m(f)-i(f)) + 4b(f)
\end{align*}
This concludes our proof.
\end{proof}

In the following lemma, we improve Arikushi et al.'s upper bound on the number of edges of $G_1$ that $G$ may contain, when the plane subgraph $G_0$ is obtained from a plane triangulation $T$ by removing $k$~edges, under the assumption that $T$ may contain non-homotopic parallel edges (but no self-loops), and that  each face $f \in F_0$ of $G_0$ is \good.  Let $t(f)$ be the minimum number of edges that must be removed from $T$ to obtain $f$. Similar to Arikushi et al., we preliminarily observe that a face $f$ of $G_0$ with $t(f) = 0$ cannot contain edges of $G_1$ in $G$. If $t(f) = 1$, the only two possible configurations for face $f$ are illustrated  in Figs.~\ref{fig:smallFaces1} and~\ref{fig:smallFaces2}. In both cases, face $f$ can contain at most two crossing edges.  If $t(f) = 2$, the only three possible configurations for face $f$ are illustrated  in Figs.~\ref{fig:smallFaces3}--\ref{fig:smallFaces5}. Then, face $f$ can contain at most five crossing edges. Let $F_0^1$ and $F_0^2$ be the set of faces of $G_0$ that can be obtained from triangulation $T$ by removing $1$ and $2$ edges, respectively, that is, $F_0^1 = \{f \in F_0; t(f)= 1\}$ and $F_0^2 = \{f \in F_0; t(f)= 2\}$. By Lemma~\ref{lem:crossed-in-f} and the previous observations, we have
\begin{equation}
\label{eq:crossed-good}
|E_1| \leq 2|F_0^1| + 5|F_0^2| + \sum_{f \in F_0; t(f) > 2}(2d(f) -2 m(f) + 2i(f) +4b(f)- 8) 
\end{equation}

\begin{figure}[t]
	\centering
	\subfloat[\label{fig:smallFaces1}{$4,0,0,1$}]
	{\includegraphics[page=1,scale=1.1]{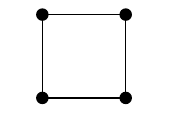}}
	\hfil 	
	\subfloat[\label{fig:smallFaces2}{$3,1,0,2$}]
	{\includegraphics[page=3,scale=1.1]{smallFaces}}
	\hfil 	
	\subfloat[\label{fig:smallFaces3}{$5,0,0,1$}]
	{\includegraphics[page=4,scale=1.1]{smallFaces}}
	\hfil 	
	\subfloat[\label{fig:smallFaces4}{$4,1,0,2$}]
	{\includegraphics[page=2,scale=1.1]{smallFaces}}
	\hfil 	
	\subfloat[\label{fig:smallFaces5}{$3,0,1,1$}]
	{\includegraphics[page=7,scale=1.1]{smallFaces}}
	
	\subfloat[\label{fig:smallFaces6}{$6,0,0,1$}]
	{\includegraphics[page=8,scale=1.1]{smallFaces}}
	\hfil 	
	\subfloat[\label{fig:smallFaces7}{$5,1,0,2$}]
	{\includegraphics[page=5,scale=1.1]{smallFaces}}
	\hfil 	
	\subfloat[\label{fig:smallFaces8}{$4,0,1,1$}]
	{\includegraphics[page=6,scale=1.1]{smallFaces}}
	\hfil 	
	\subfloat[\label{fig:smallFaces9}{$4,2,0,3$}]
	{\includegraphics[page=9,scale=1.1]{smallFaces}}
	\hfil 	
	\subfloat[\label{fig:smallFaces10}{$4,2,0,3$}]
	{\includegraphics[page=10,scale=1.1]{smallFaces}}
	\hfil 	
	\subfloat[\label{fig:smallFaces11}{$4,2,0,3$}]
	{\includegraphics[page=11,scale=1.1]{smallFaces}}
	\caption{All bounded faces that can be obtained from $T$ by removing 
	(a)--(b)~$1$~edge, 
	(c)--(e)~$2$~edges, 
	(f)--(k)~$3$~edges.
	The caption of each subfigure indicates the values of 
	$(d(f),m(f),i(f),b(f))$.}
	\label{fig:smallFaces}
\end{figure}

In the following lemma, we prove that a slight overestimation of the right-hand side of Eq.~\ref{eq:crossed-good} is upper-bounded by $\frac{8}{3}k$, which clearly implies that $|E_1| \leq \frac{8}{3}k$.

\begin{lemma}\label{lem:crossed-better}
If $G_0$ is obtained from triangulation $T$ by removing $k$ edges, then:
\begin{equation}
\frac{8}{3}|F_0^1| + \frac{16}{3}|F_0^2| + \sum_{f \in F_0; t(f) > 2}(2d(f) -2 m(f) + 2i(f) +4b(f)- 8)  \leq \frac{8}{3}k
\label{eq:crossed-better}
\end{equation}
\end{lemma}
\begin{proof}
Our proof is by induction on $k$ and is similar to the corresponding one of Arikushi et al.~(Lemma~5 in~\cite{DBLP:journals/comgeo/ArikushiFKMT12}). In contrast to their proof, we assume that $G_0$ is obtained from triangulation $T$ by removing edges in a certain order. In particular, we want to avoid the case in which the removal of an edge $e$ results in merging two faces $f_1$ and $f_2$ such that $t(f_1),t(f_2)  \geq 1$ (refer to Case C.\ref{C.arikushi4} in Section~\ref{sec:preliminaries}).  
We guarantee this property as follows. Consider the subgraph $D$ of the dual of $T$ induced by the edges that are dual to those that we have to remove to obtain $G_0$. We remove the edges in the order in which their dual edges appear in a BFS traversal of each connected component of $D$. In this way, every inter-level edge in the BFS traversal corresponds to removing an edge that is incident to a triangular face (not visited yet), while each intra-level edge corresponds to removing a bridge from a face that has been created by previously removed edges. In both cases, we avoid merging two faces $f_1$ and $f_2$ such that $t(f_1),t(f_2)  \geq 1$.

Denote by $\tau(G_0)$ the left-hand side of Eq.~\ref{eq:crossed-better}. In the base of the induction, $k=0$ holds. In this case, graph $G_0$ coincides with triangulation $T$ and thus $\tau(G_0)=0$. In the induction hypothesis, we assume that the lemma holds for $k \geq 0$, and we prove that it also holds for $k' = k+1$.
 
Let $G_0'$ be a plane graph obtained from $T$ by removing $k'$ edges, and let $G_0$ be the plane graph obtained from $T$ by removing the same $k'$ edges, except for the last one, which we call $(u,v)$. For $G_0$, by induction, it holds that $\tau(G_0) \leq \frac{8}{3}k$. We consider the following cases:
\begin{enumerate}[C.1]
\item \label{mainCase1} Edge $(u,v)$ is a bridge of a face $f$ in $G_0$ such that $t(f) \geq 3$.  Let $f'$ be the face of $G_0'$ that is obtained by the removal of $(u,v)$. Note that $t(f') \geq 4$. Since $(u,v)$ is a biconnected component of $f$, it holds that $b(f') = b(f) - 1$. Since $(u,v)$ is a bridge, it also holds that $d(f') = d(f)$. To establish the values of $m(f')$ and $i(f')$, we observe that $u$, or $v$, or both may become isolated vertices of $G_0^\prime$ after the removal of $(u,v)$. We study these cases separately. 
\begin{figure}[t]
	\centering
	\subfloat[\label{fig:mainCases1}{C.\ref*{mainCase1a}}]
	{\includegraphics[page=1,scale=1.1]{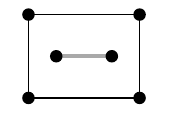}}
	\hfil 	
	\subfloat[\label{fig:mainCases2}{C.\ref*{mainCase1b}}]
	{\includegraphics[page=2,scale=1.1]{maincases}}
	\hfil 	
	\subfloat[\label{fig:mainCases3}{C.\ref*{mainCase1c}}]
	{\includegraphics[page=3,scale=1.1]{maincases}}
	\hfil 	
	\subfloat[\label{fig:mainCases4}{C.\ref*{mainCase2a}}]
	{\includegraphics[page=4,scale=1.1]{maincases}}
	\hfil 	
	\subfloat[\label{fig:mainCases5}{C.\ref*{mainCase2b}}]
	{\includegraphics[page=5,scale=1.1]{maincases}}
	\hfil 	
	\subfloat[\label{fig:mainCases6}{C.\ref*{mainCase2b}}]
	{\includegraphics[page=6,scale=1.1]{maincases}}
	\caption{Illustrations of Cases~C.\ref{mainCase1} and~C.\ref{mainCase2}. Edge $(u,v)$ is gray-colored.}
	\label{fig:mainCases}
\end{figure}
\begin{enumerate}[a)]
\item \label{mainCase1a} Both $u$ and $v$ become isolated vertices in $G_0^\prime$; see Fig.~\ref{fig:mainCases1}. Then $m(f') = m(f)$ and $i(f') = i(f) + 2$. Since $2d(f') - 2 m(f') + 2i(f') + 4b(f')- 8 = 2d(f) - 2m(f) + 2(i(f)+2) +4(b(f)-1) - 8 = 2d(f) - 2 m(f) + 2i(f) + 4b(f)- 8 $, it follows that $\tau(G_0')=\tau(G_0) \leq \frac{8}{3}k < \frac{8}{3}k'$. 
\item \label{mainCase1b} Exactly one of $u$ and $v$, say $v$, becomes an isolated vertex in $G_0^\prime$; see Fig.~\ref{fig:mainCases2}. Then $m(f') = m(f)-1$ and $i(f') = i(f) + 1$. Since $2d(f') - 2 m(f') + 2i(f') + 4b(f')- 8 = 2d(f) - 2(m(f)-1) + 2(i(f)+1) +4(b(f)-1) - 8 = 2d(f) - 2 m(f) + 2i(f) + 4b(f)- 8 $, it follows that $\tau(G_0')=\tau(G_0) \leq \frac{8}{3}k < \frac{8}{3}k'$. 
\item \label{mainCase1c} Neither $u$ nor $v$ becomes an isolated vertex in $G_0^\prime$; see Fig.~\ref{fig:mainCases3}. Then $m(f') = m(f)-2$ and $i(f') = i(f)$. Since $2d(f') - 2 m(f') + 2i(f') + 4b(f')- 8 = 2d(f) - 2(m(f)-2) + 2i(f) + 4(b(f)-1) - 8 = 2d(f) - 2 m(f) + 2i(f) + 4b(f)- 8 $, it follows that $\tau(G_0')=\tau(G_0) \leq \frac{8}{3}k < \frac{8}{3}k'$. 
\end{enumerate}

\item \label{mainCase2} The removal of $(u,v)$ merges a triangular face $\Delta$ (that is, $t(\Delta)=0$) with an adjacent face $f$ of~$G_0$ with $t(f) \geq 3$ into a face $f'$ of $G_0'$. Note that $t(f') \geq 4$.
We consider two cases:
\begin{enumerate}[a)]
\item \label{mainCase2a} Faces $\Delta$ and $f$ share only edge $(u,v)$; see Fig.~\ref{fig:mainCases4}. Then $d(f') = d(f) + 1$, $m(f') = m(f)$, $b(f') = b(f)$, $i(f') = i(f)$. Since $2d(f') - 2 m(f') + 2i(f') + 4b(f')- 8 = 2(d(f)+1) - 2m(f) + 2i(f) +4b(f) - 8 = 2d(f) - 2 m(f) + 2i(f) + 4b(f)- 8 + 2$, it follows that $\tau(G_0')=\tau(G_0) + 2 \leq \frac{8}{3}k + 2 < \frac{8}{3}k'$. 
\item \label{mainCase2b} Faces $\Delta$ and $f$ share at least two edges; see Figs.~\ref{fig:mainCases5} and~\ref{fig:mainCases6}. By removing $(u,v)$, the number of occurrences of the third vertex $v^\prime$ of $\Delta$ increases by one and the number of biconnected components increases by one. Then $d(f') = d(f)$, $m(f') = m(f)+1$, $b(f') = b(f)+1$, $i(f') = i(f)$. Since $2d(f') - 2 m(f') + 2i(f') + 4b(f')- 8 = 2d(f) - 2(m(f)+1) + 2i(f) +4(b(f)+1) - 8 = 2d(f) - 2 m(f) + 2i(f) + 4b(f)- 8 + 2$, it follows that $\tau(G_0')=\tau(G_0) + 2 \leq \frac{8}{3}k + 2 < \frac{8}{3}k'$.
\end{enumerate}

\item The removal of $(u,v)$ yields a face $f'$ of $G_0'$ with $t(f') \in \{1,2, 3\}$. Note that in the previous cases $t(f') \geq 4$. So, if we rule out this case, then the proof follows. We consider two cases, which correspond to Cases~C.\ref{mainCase1} and~C.\ref{mainCase2} for smaller faces, respectively. 
\begin{enumerate}[a)]
\item Face $f'$ is obtained by removing a bridge from a face $f$. Hence, $t(f) = t(f^\prime) - 1$ and $f'$ is disconnected. 
Observe that if $t(f') = 1$, then face $f^\prime$ is not disconnected as can be seen from Figs.~\ref{fig:smallFaces1} and~\ref{fig:smallFaces2}. Therefore, $t(f') \geq 2$ holds in this subcase.
\item Face $f'$ is obtained by merging a face $f$ with a triangular face $\Delta$. Hence, $t(f) = t(f^\prime) - 1$ holds. Since $\Delta$ is triangular, we observe that it does not contribute to $\tau(G_0)$. 

\end{enumerate}
In both cases, the face $f$ that is eliminated in order to create face $f'$ is such that $t(f) = t(f')-1$. We observe that $\tau(G_0')$ is equal to $\tau(G_0)$, plus the contribution of $f'$ to $\tau(G_0')$, minus the contribution of $f$ to $\tau(G_0)$. More precisely:
If $t(f') = 1$, then $\tau(G_0') = \tau(G_0) + \frac{8}{3} - 0 \leq \frac{8}{3}k + \frac{8}{3} =  \frac{8}{3}k'$; see Figs.~\ref{fig:smallFaces1}-\ref{fig:smallFaces2}. If $t(f') = 2$, then $\tau(G_0') = \tau(G_0) + \frac{16}{3} - \frac{8}{3} \leq \frac{8}{3}k + \frac{8}{3} =  \frac{8}{3}k'$; see Figs.~\ref{fig:smallFaces3}-\ref{fig:smallFaces5}. 
Otherwise, $t(f')=3$; see Figs.~\ref{fig:smallFaces6}-\ref{fig:smallFaces11}. This implies that $\tau(G_0') \leq \tau(G_0) + (2d(f') -2 m(f') + 2i(f') +4b(f')- 8) - \frac{16}{3}$. It is easy to verify that $2d(f') -2 m(f') + 2i(f') +4b(f')- 8 \leq 8$ holds for each of the cases shown in Figs.~\ref{fig:smallFaces6}-\ref{fig:smallFaces11}. Hence, $\tau(G_0') \leq \tau(G_0) + \frac{8}{3} \leq \frac{8}{3}k + \frac{8}{3} =  \frac{8}{3} k'$.  
\end{enumerate}
This concludes the proof.\end{proof}

By following a counting similar to Arikushi et al. we obtain a bound on the maximum number of edges of a \rac{1} graph with $n$ vertices, when all the faces of $G_0$ are \good. Since planar graphs have at most $3n-6$ edges even in the presence of non-homotopic parallel edges, the bound is obtained when $7n-14-k = 3n-6-k+\frac{8}{3}k$, that is, $k=\frac{3}{2}(n-2)$. This directly implies that in this case $|E| \leq 5.5n - 11$. 

In the following, we prove that it is not a loss of generality to assume that all faces of $G_0$ are \good, as otherwise we can augment our graph by adding only crossing-free edges to $G$ (not necessarily drawn with one bend but rather as curves), in such a way that every face of $G_0$ becomes \good. Recall that we denote by $G^\prime$ the planarization of drawing $\Gamma$ of $G$.

Assume that there exists a face of $G_0$ that is not \good. Hence, there exist at least two edges belonging to $G_0$ which are incident to the same face $f^\prime$ in $G^\prime$. If $f^\prime$ consists exclusively of edges of $G_0$, then we triangulate $f^\prime$. Otherwise, we traverse the facial walk of $f^\prime$ starting from any dummy vertex of $f^\prime$ and we connect by a crossing-free edge the first occurring vertex that is incident to an edge of $G_0$ with the last occurring vertex that is also incident to an edge of $G_0$. This implies that one of the two faces into which $f^\prime$ is split contains only one crossing-free edge, namely the newly added edge. Note that, in both cases, it is always possible to add the described edges, since we do not require them to be drawn with one bend. Since in both cases, we split a face into smaller faces, this process eventually terminates. At the end, each face is either a triangle of crossing-free edges or contains at most one crossing-free edge. Hence, it is indeed  not a loss of generality to assume that all faces of $G_0$ are \good.

We remark that the aforementioned procedure may result in parallel edges or self-loops, which are however non-homotopic by construction. In particular, a self-loop may appear, when the first and the last occurring vertices in the facial walk are identified and form a cut-vertex of $G$. Note that while Lemma~\ref{lem:crossed-better} allows non-homotopic parallel edges, it does not allow self-loops. Hence, for self-loops we need to use a different approach. Consider self-loop $s$. As already mentioned, $s$ is incident to a cut-vertex of $G$ and encloses a part of $\Gamma$, which we assume not to contain any other self-loop. Let $H_1$ and $H_2$ be the subgraphs of $G$ that are induced by the vertices of $G$ that are in the interior and the exterior of~$s$, respectively. Denote by $n_1$ and $n_2$ the number of vertices of $H_1$ and $H_2$, respectively, and by $m_1$ and $m_2$ their corresponding number of edges. Observe that $n=n_1+n_2-1$. Note that edge $s$ is accounted neither in $H_1$ nor in $H_2$. By induction, we may assume that $m_1 \leq 5.5n_1-11$ and $m_2 \leq 5.5n_2-11$. Hence, graph $G$ (including $s$) contains at most $5.5(n_1+n_2)-22 + 1 = 5.5n - 15.5 \leq 5.5n-11$ edges. This implies that the upper bound holds even in the presence of self-loops.

We are now ready to state the main theorem of this section.

\begin{theorem}
Every $n$-vertex \rac{1} graph has at most $5.5n-11$ edges.
\end{theorem}
%
%
%
%

%

\section{An Improved Lower Bound}
\label{sec:lower-bound}

In this section, we present an improved lower bound for the number of edges of \rac{1} graphs. Our construction is partially inspired by the corresponding lower bound constructions of $2$-planar graphs~\cite{DBLP:conf/compgeom/Bekos0R17} and fan-planar graphs~\cite{DBLP:journals/corr/KaufmannU14} with maximum density.

\begin{theorem}
There exists infinitely many $n$-vertex \rac{1} graphs with exactly $5n-10$ edges.
\end{theorem}
\begin{proof}
A central ingredient in our lower bound construction is the dodecahedral graph; see Fig.~\ref{fig:lowerBound1}. This graph admits a straight-line planar drawing in which the outer face is a regular pentagon, and the inner faces can be partitioned into three sets, based on their shape. Namely, the innermost face (shaded in gray in Fig.~\ref{fig:lowerBound1}) is again a regular pentagon, vertically mirrored with respect to the outer one; also, all the faces adjacent to the innermost face have the same shape, which we will describe more precisely later, and the same holds for all the faces adjacent to the outer face. In particular, the drawing of each face is symmetric with respect to the line that is perpendicular to one of its sides (whose length is denoted by $a$ in Fig.~\ref{fig:lowerBound2}) and passes through its opposite vertex (denoted by $A$ in Fig.~\ref{fig:lowerBound2}). Adopting the notation scheme of Fig.~\ref{fig:lowerBound2}, in the following we provide  values for the angles and side length ratios to fully describe the shapes of the faces adjacent to the innermost face and to the outer face; for an illustration, refer to Fig.~\ref{fig:lowerBound1}.

\begin{figure}[t]
	\centering
	\begin{minipage}[b]{0.59\textwidth}
	  \subfloat[\label{fig:lowerBound1}]
	  {\includegraphics[page=1,width=\textwidth]{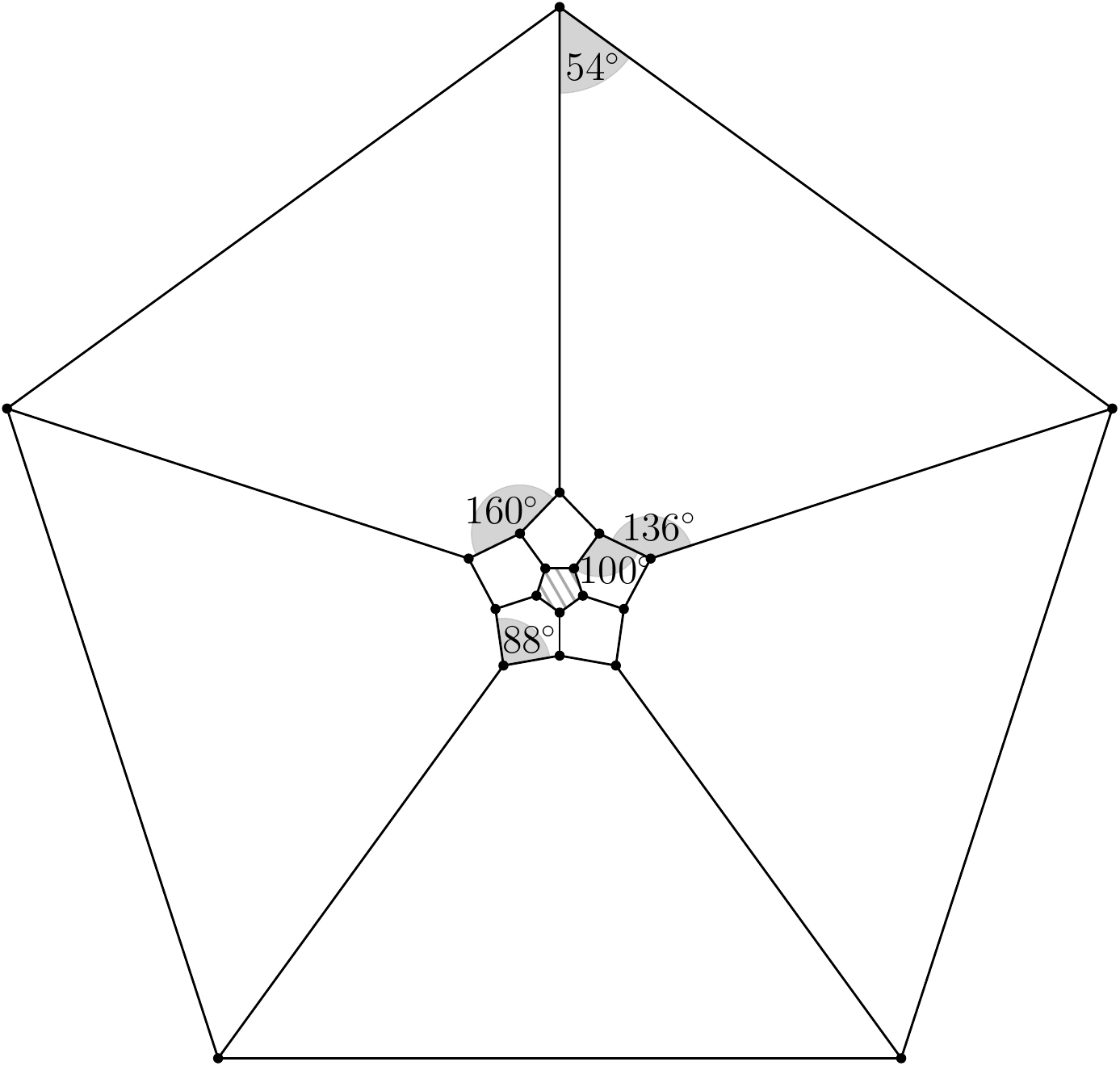}}
	\end{minipage}	
	\hfil 
	\begin{minipage}[b]{0.3\textwidth}
	  \subfloat[\label{fig:lowerBound2}]
	  {\includegraphics[page=7,width=0.9\textwidth]{lowerBound}}
	  
	  \subfloat[\label{fig:lowerBound3}]
	  {\includegraphics[page=8,width=0.9\textwidth]{lowerBound}}
	\end{minipage}	
	\caption{Illustrations for the lower bound construction: 
	(a)~the dodecahedral graph, 
	(b)~angles and edge lengths, and
	(c)~crossing configuration.}
	\label{fig:lowerBound}
\end{figure}

\begin{enumerate}[(i)]
\item The five faces adjacent to the innermost face are realized such that the side of length $a$ is incident to the inner face. Angles $\alpha$ and $\beta$ are $88^\circ$ and $100^\circ$, respectively. In addition, side-length $b$ is $1.5$ times the side-length $a$.

\item The five faces adjacent to the outer face are realized such that the side of length $a$ is incident to the outer face. Angles $\alpha$ and $\gamma$ are $160^\circ$ and $54^\circ$, respectively. In addition, side-length $b$ is $8.5$ times side-length $c$.
\end{enumerate}

Consider two copies $D_1$ and $D_2$ of this drawing of the dodecahedral graph. Since both the innermost face of $D_1$ and the outer face of $D_2$ are drawn as regular pentagons, after scaling the drawing $D_2$ uniformly and mirroring it vertically, we can construct a drawing of a larger graph by identifying the innermost face of $D_1$ with the outer face of $D_2$. This process can be clearly repeated arbitrarily many times. The result is a graph family such that every member of this family admits a straight-line planar drawing, in which each face has one of the shapes described above.

\begin{figure}[t]
	\centering
	\subfloat[\label{fig:lowerBound-innermost}]
	{\includegraphics[page=2,width=.3\textwidth]{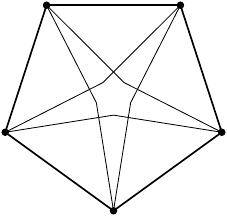}}
	\hfil
	\subfloat[\label{fig:lowerBound-outermost}]
	{\includegraphics[page=5,width=.3\textwidth]{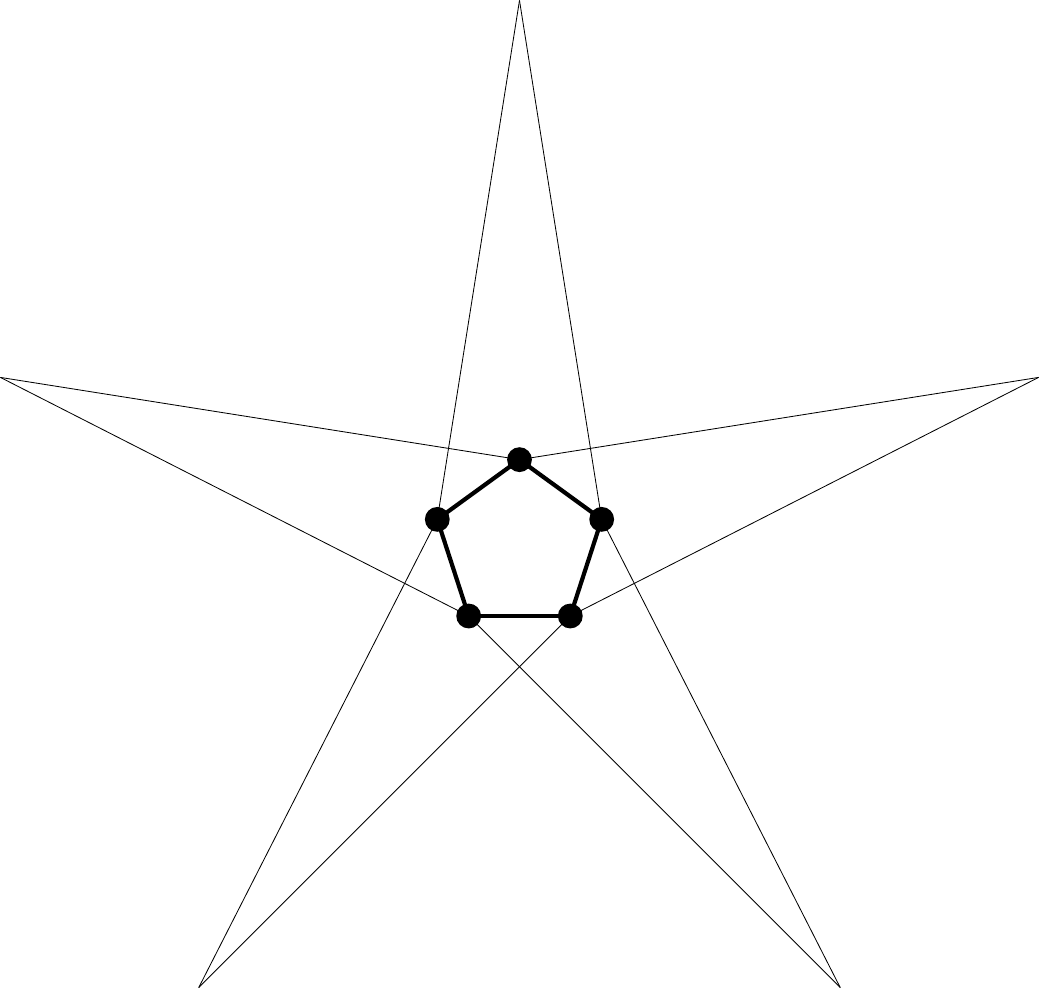}}
	\caption{Chords inside 
	(a)~the innermost face, and 
	(b)~the outer face.}
	\label{fig:lowerBoundBoundaryFaces}
\end{figure}

For our lower bound construction, we add five chords in the interior of each face of every member of the above family. Hence, the five vertices that are incident to each face induce a complete graph $K_5$. In the following, we describe how to draw such chords in the interior of each of the aforementioned faces, based on their shape, so that the resulting drawing is \rac{1}. For an illustration of the configuration of the crossing edges in each of these faces refer to Fig.~\ref{fig:lowerBound3}; we will formally define angles $\alpha_1,\beta_1,\beta_2,\gamma_1,\gamma_2$ shortly. Observe that all edges and the formed angles are symmetric with respect to the line through vertex $A$ that is perpendicular to $C_1C_2$. Also, for every three vertices $u$, $w$, and $v$ that are consecutive along the boundary of the face, the chord $(u,v)$ will cross both chords incident to $w$, making a bend between these two crossings. In the following, we provide values for the angles $\alpha_1,\beta_1,\beta_2,\gamma_1,\gamma_2$ to fully describe the configurations of the crossing edges.
\begin{enumerate}[(i)]
\item For the innermost face, $\alpha_1=\beta_1=\beta_2=\gamma_1=\gamma_2=45^\circ$ holds; refer to Fig.~\ref{fig:lowerBound-innermost}. 
\item For the outer face, $\alpha_1=\beta_1=\beta_2=\gamma_1=\gamma_2=45^\circ$ holds; refer to Fig.~\ref{fig:lowerBound-outermost}. 
\item For the five faces neighboring the innermost face, $\alpha_1=40^\circ$, $\beta_1=30^\circ$, $\beta_2=50^\circ$, $\gamma_1 =45^\circ$ and $\gamma_2 =60^\circ$ holds; refer to \arxapp{Fig.~\ref{fig:lowerBoundFace1} in Appendix~\ref{app:lowerBound}}{\cite{arxivVersion}}. 
\item For the five faces neighboring the outer face, $\alpha_1 =47.5^\circ$, $\beta_1 =85^\circ$, $\beta_2 =42.5^\circ$, $\gamma_1 =45^\circ$ and $\gamma_2 =5^\circ$ holds; refer to \arxapp{Fig.~\ref{fig:lowerBoundFace2} in Appendix~\ref{app:lowerBound}}{\cite{arxivVersion}}. 
\end{enumerate}

It follows that each graph in the family admits a \rac{1} drawing. Let~$G_n$ be such a graph with $n$ vertices. Next, we discuss the exact number of edges of graph $G_n$.  Since the crossing-free edges of $G_n$ form a planar graph, whose faces are all of length $5$, it follows by Euler's formula that this graph has $\frac{5}{3}(n-2)$ edges and $\frac{2}{3}(n-2)$ faces. Since each of these faces contains five chords, the number of edges of $G_n$ is $\frac{5}{3}(n-2) + 5\cdot \frac{2}{3}(n-2)= 5n-10$, and the statement follows. 
\end{proof}

\section{Conclusions}\label{sec:conclusions}

In this paper, we improved the previously best lower and upper bounds on the number of edges of \rac{1} graphs. The gap between our lower and upper bound is approximately $n/2$. A future challenge will be to further narrow this gap. We conjecture that an $n$-vertex \rac{1} graph cannot have more than $5n-10$ edges (as it is the case for several other classes of beyond planar graphs; see e.g.~\cite{DBLP:conf/gd/BaeBCEEGHKMRT17,DBLP:journals/corr/KaufmannU14,PachT97}). Significantly more difficult seems to be the problem of improving the current best lower and upper bounds on the number of edges of \rac{2} graphs, where the gap is significantly wider (approx., $67n$). Closely connected are also complexity related questions; in particular, the characterization and recognition of $1$- and \rac{2} graphs are still open.

\paragraph{Acknowlegdment.} This project was supported by DFG grant KA812/18-1.

\bibliographystyle{splncs03}
\bibliography{abbrv,references}

\begin{thebibliography}{10}
\providecommand{\url}[1]{\texttt{#1}}
\providecommand{\urlprefix}{URL }

\bibitem{Ackerman09}
Ackerman, E.: On the maximum number of edges in topological graphs with no four
  pairwise crossing edges. Discrete Comput. Geom.  41(3),  365--375 (2009)

\bibitem{DBLP:journals/corr/Ackerman15}
Ackerman, E.: On topological graphs with at most four crossings per edge. CoRR
  abs/1509.01932 (2015)

\bibitem{DBLP:journals/jgaa/AckermanKV18}
Ackerman, E., Keszegh, B., Vizer, M.: On the size of planarly connected
  crossing graphs. J. Graph Algorithms Appl.  22(1),  11--22 (2018)

\bibitem{AckermanT07}
Ackerman, E., Tardos, G.: On the maximum number of edges in quasi-planar
  graphs. J. Comb. Theory, Series A  114(3),  563--571 (2007)

\bibitem{AgarwalAPPS97}
Agarwal, P.K., Aronov, B., Pach, J., Pollack, R., Sharir, M.: Quasi-planar
  graphs have a linear number of edges. Combinatorica  17(1),  1--9 (1997)

\bibitem{DBLP:journals/jgaa/AngeliniCDFBKS11}
Angelini, P., Cittadini, L., Didimo, W., Frati, F., {Di Battista}, G.,
  Kaufmann, M., Symvonis, A.: On the perspectives opened by right angle
  crossing drawings. J. Graph Algorithms Appl.  15(1),  53--78 (2011)

\bibitem{DBLP:journals/jgaa/ArgyriouBS12}
Argyriou, E.N., Bekos, M.A., Symvonis, A.: The straight-line {RAC} drawing
  problem is {NP}-hard. J. Graph Algorithms Appl.  16(2),  569--597 (2012)

\bibitem{DBLP:journals/comgeo/ArikushiFKMT12}
Arikushi, K., Fulek, R., Keszegh, B., Moric, F., T{\'{o}}th, C.D.: Graphs that
  admit right angle crossing drawings. Comput. Geom.  45(4),  169--177 (2012)

\bibitem{DBLP:journals/dam/BachmaierBHNR17}
Bachmaier, C., Brandenburg, F.J., Hanauer, K., Neuwirth, D., Reislhuber, J.:
  Nic-planar graphs. Discrete Applied Mathematics  232,  23--40 (2017)

\bibitem{DBLP:conf/gd/BaeBCEEGHKMRT17}
Bae, S.W., Baffier, J., Chun, J., Eades, P., Eickmeyer, K., Grilli, L., Hong,
  S., Korman, M., Montecchiani, F., Rutter, I., T{\'{o}}th, C.D.: Gap-planar
  graphs. In: Frati, F., Ma, K. (eds.) {Graph Drawing and Network
  Visualization}. LNCS, vol. 10692, pp. 531--545. Springer (2017)

\bibitem{DBLP:journals/tcs/BekosDLMM17}
Bekos, M.A., Didimo, W., Liotta, G., Mehrabi, S., Montecchiani, F.: On {RAC}
  drawings of 1-planar graphs. Theor. Comput. Sci.  689,  48--57 (2017)

\bibitem{DBLP:conf/gd/Bekos0R16}
Bekos, M.A., Kaufmann, M., Raftopoulou, C.N.: On the density of non-simple
  3-planar graphs. In: Hu, Y., N{\"{o}}llenburg, M. (eds.) Graph Drawing. LNCS,
  vol. 9801, pp. 344--356. Springer (2016)

\bibitem{DBLP:conf/compgeom/Bekos0R17}
Bekos, M.A., Kaufmann, M., Raftopoulou, C.N.: On optimal 2- and 3-planar
  graphs. In: Aronov, B., Katz, M.J. (eds.) Symposium on Computational
  Geometry. LIPIcs, vol.~77, pp. 16:1--16:16. Schloss Dagstuhl -
  Leibniz-Zentrum fuer Informatik (2017)

\bibitem{DBLP:journals/tcs/BrandenburgDEKL16}
Brandenburg, F.J., Didimo, W., Evans, W.S., Kindermann, P., Liotta, G.,
  Montecchiani, F.: Recognizing and drawing ic-planar graphs. Theor. Comput.
  Sci.  636,  1--16 (2016)

\bibitem{clwz-brdngqa-eurocg18}
Chaplick, S., Lipp, F., Wolff, A., Zink, J.: 1-bend rac drawings of nic-planar
  graphs in quadratic area. In: Mulzer, W. (ed.) Proc. 34th European Workshop
  on Computational Geometry (EuroCG'18). Berlin (2018), to appear.

\bibitem{CheongHKK15}
Cheong, O., Har{-}Peled, S., Kim, H., Kim, H.: On the number of edges of
  fan-crossing free graphs. Algorithmica  73(4),  673--695 (2015)

\bibitem{DBLP:journals/algorithmica/GiacomoDEL14}
{Di Giacomo}, E., Didimo, W., Eades, P., Liotta, G.: 2-layer right angle
  crossing drawings. Algorithmica  68(4),  954--997 (2014)

\bibitem{DBLP:journals/mst/GiacomoDLM11}
{Di Giacomo}, E., Didimo, W., Liotta, G., Meijer, H.: Area, curve complexity,
  and crossing resolution of non-planar graph drawings. Theory Comput. Syst.
  49(3),  565--575 (2011)

\bibitem{DBLP:journals/ipl/DidimoEL10}
Didimo, W., Eades, P., Liotta, G.: A characterization of complete bipartite
  {RAC} graphs. Inf. Process. Lett.  110(16),  687--691 (2010)

\bibitem{DBLP:journals/tcs/DidimoEL11}
Didimo, W., Eades, P., Liotta, G.: Drawing graphs with right angle crossings.
  Theor. Comput. Sci.  412(39),  5156--5166 (2011)

\bibitem{DBLP:journals/corr/abs-1804-07257}
Didimo, W., Liotta, G., Montecchiani, F.: A survey on graph drawing beyond
  planarity. CoRR  abs/1804.07257 (2018)

\bibitem{DBLP:journals/dam/EadesL13}
Eades, P., Liotta, G.: Right angle crossing graphs and 1-planarity. Discrete
  Applied Mathematics  161(7-8),  961--969 (2013)

\bibitem{FoxPS13}
Fox, J., Pach, J., Suk, A.: The number of edges in k-quasi-planar graphs.
  {SIAM} J. Discrete Math.  27(1),  550--561 (2013)

\bibitem{DBLP:conf/wg/HongN15}
Hong, S., Nagamochi, H.: Testing full outer-2-planarity in linear time. In:
  Mayr, E.W. (ed.) {WG}. LNCS, vol. 9224, pp. 406--421. Springer (2015)

\bibitem{Shonan2016}
Hong, S., Tokuyama, T.: Algorithmics for beyond planar graphs. NII Shonan
  Meeting Seminar 089 (November 27 - December 1 2016)

\bibitem{DBLP:conf/apvis/Huang07}
Huang, W.: Using eye tracking to investigate graph layout effects. In: Hong,
  S., Ma, K. (eds.) {APVIS}. pp. 97--100. {IEEE} Computer Society (2007)

\bibitem{DBLP:journals/vlc/HuangEH14}
Huang, W., Eades, P., Hong, S.: Larger crossing angles make graphs easier to
  read. J. Vis. Lang. Comput.  25(4),  452--465 (2014)

\bibitem{Dagstuhl2016}
Kaufmann, M., Kobourov, S., Pach, J., Hong, S.: Beyond planar graphs:
  Algorithmics and combinatorics. Dagstuhl Seminar 16452 (November 6-11 2016)

\bibitem{DBLP:journals/corr/KaufmannU14}
Kaufmann, M., Ueckerdt, T.: The density of fan-planar graphs. CoRR
  abs/1403.6184 (2014)

\bibitem{SoCG2017}
Liotta, G.: Graph drawing beyond planarity: Some results and open problems.
  SoCG Week, Invited talk (July 4th 2017)

\bibitem{PachRTT06}
Pach, J., Radoi\v{c}i{\'{c}}, R., Tardos, G., T{\'{o}}th, G.: Improving the
  crossing lemma by finding more crossings in sparse graphs. Discrete {\&}
  Computational Geometry  36(4),  527--552 (2006)

\bibitem{PachT97}
Pach, J., T{\'o}th, G.: Graphs drawn with few crossings per edge. Combinatorica
   17(3),  427--439 (1997)

\bibitem{DBLP:journals/iwc/Purchase00}
Purchase, H.C.: Effective information visualisation: a study of graph drawing
  aesthetics and algorithms. Interacting with Computers  13(2),  147--162
  (2000)

\bibitem{DBLP:journals/ese/PurchaseCA02}
Purchase, H.C., Carrington, D.A., Allder, J.: Empirical evaluation of
  aesthetics-based graph layout. Empirical Software Engineering  7(3),
  233--255 (2002)

\bibitem{MR0187232}
Ringel, G.: Ein {S}echsfarbenproblem auf der {K}ugel. Abh. Math. Sem. Univ.
  Hamb.  29,  107--117 (1965)

\end{thebibliography}

\arxapp{\clearpage

\appendix

\section{Additional Figures for the Lower Bound Construction}
\label{app:lowerBound}

\begin{figure}[htb!]
	\centering
	\includegraphics[page=3,width=0.7\textwidth]{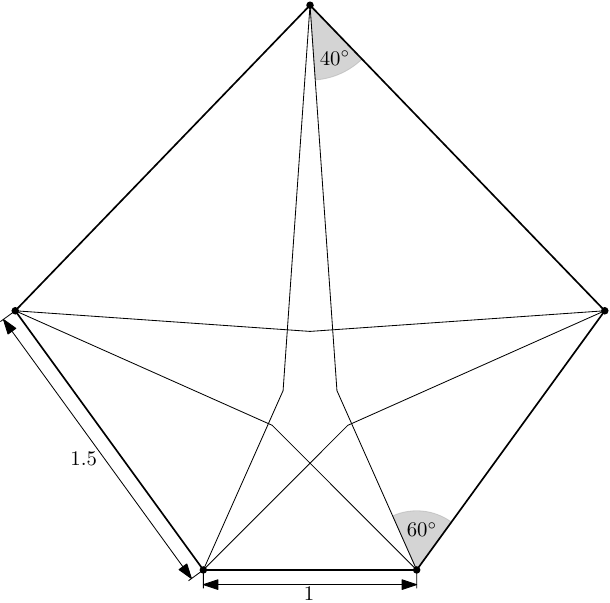}
	\caption{Chords inside faces neighboring the innermost face.}
	\label{fig:lowerBoundFace1}
\end{figure}

\begin{figure}[htb!]
	\centering
	\includegraphics[page=4,width=0.7\textwidth]{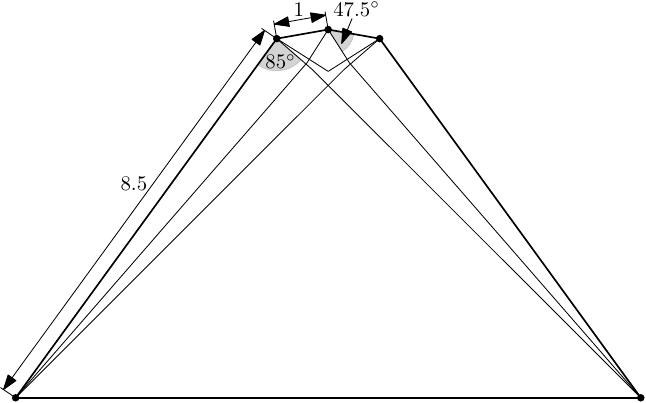}
	\caption{Chords inside faces neighboring the outer face.}
	\label{fig:lowerBoundFace2}
\end{figure}}

\end{document}